\documentclass[letterpaper, 10 pt, conference]{ieeeconf}  

\IEEEoverridecommandlockouts                              
\usepackage{amsmath} 
\usepackage{amsthm}

\usepackage{amsfonts}
\usepackage{amssymb}  
\usepackage{mathtools} 
\usepackage{mathrsfs}
\usepackage{balance}
\usepackage{url}
\usepackage{hyperref}
\DeclareMathAlphabet{\mathpzc}{OT1}{pzc}{m}{it}

\newcommand{\R}{\mathbb{R}}

\theoremstyle{definition}
\newtheorem{proposition}{Proposition}
\newtheorem{definition}{Definition}

\newtheorem{lemma}{Lemma}
\newtheorem{theorem}{Theorem}
\newtheorem{remark}{Remark}

\usepackage{cleveref}

\usepackage[ruled,linesnumbered]{algorithm2e}

\DeclareMathOperator*{\argmin}{arg\,min}
\DeclareMathOperator*{\argmax}{arg\,max}

\newcommand{\doi}[1]{\href{http://dx.doi.org/#1}{\normalsize{\textsc{doi:}}~\nolinkurl{#1}}}
\newcommand{\arxiv}[1]{\href{http://arxiv.org/abs/#1}{\normalsize{\textsc{arxiv:}}~\nolinkurl{#1}}}
\usepackage{color}

\usepackage{comment}
\newcommand{\HRule}{\noindent\rule{\linewidth}{0.1mm}\newline}

\renewcommand{\epsilon}{\varepsilon}
\renewcommand{\phi}{\varphi}

\newcommand{\LL}{\mathcal{L}}
\newcommand{\Safe}{\mathcal{S}}
\newcommand{\traj}{\psi}
\newcommand{\trajsimple}{\mathrm{x}}

\usepackage[dvipsnames]{xcolor}

\newcommand{\shnote}[1]%
    {\textcolor{magenta}{ #1}}
\newcommand{\cbf}{B_{\gamma}}
\newcommand{\strategy}{\xi_{d}}
\newcommand{\Strategy}{\Xi}
\newcommand{\UU}{\mathcal{U}}
\newcommand{\tpdelta}{t\!+\!\delta}
\newcommand{\topdelta}{t_{0}\!+\!\delta}
\newcommand{\xntn}{x_{0},t_{0}}
\newcommand{\tT}{t}

\newcommand{\ctrl}{u}
\newcommand{\dstb}{d}
\newcommand{\state}{\trajsimple}
\newcommand{\CBVF}{CBVF}

\newcommand{\cset}{U}
\newcommand{\cfset}{\UU}
\newcommand{\dset}{D}
\newcommand{\dfset}{\mathcal{D}}
\newcommand{\lf}{l}
\newcommand{\VV}{V}
\newcommand{\JJ}{J}
\newcommand{\BB}{B}
\newcommand{\BV}{B}
\newcommand{\BVg}{\BV_{\gamma}}

\title{\LARGE \bf
Robust Control Barrier--Value Functions for Safety-Critical Control
}

\author{Jason J. Choi, Donggun Lee, Koushil Sreenath, Claire J. Tomlin, and Sylvia L. Herbert
\thanks{This research is supported in part by the DARPA Assured Autonomy program, National Science Foundation Grant CMMI-1931853 and UCSD. Jason Choi, Donggun Lee, Koushil Sreenath, and Claire Tomlin are with University of California, Berkeley, and Sylvia Herbert is with University of California, San Diego. Contact info: \{jason.choi, donggun\_lee, koushils, tomlin\}@berkeley.edu,
sherbert@ucsd.edu}%
}

\begin{document}

\renewcommand{\baselinestretch}{1}
\maketitle
\thispagestyle{empty}
\pagestyle{empty}

\begin{abstract}
This paper works towards unifying two popular approaches in the safety control community: Hamilton-Jacobi (HJ) reachability and Control Barrier Functions (CBFs). HJ Reachability has methods for direct construction of value functions that provide safety guarantees and safe controllers, however the online implementation can be overly conservative and/or rely on chattering bang-bang control. The CBF community has methods for safe-guarding controllers in the form of point-wise optimization using quadratic programs (CBF-QP), where the CBF-based safety certificate is used as a constraint. However, finding a valid CBF for a general dynamical system is challenging. This paper unifies these two methods by introducing a new reachability formulation inspired by the structure of CBFs to construct a Control Barrier-Value Function (CBVF). We verify that CBVF is a viscosity solution to a novel Hamilton-Jacobi-Isaacs Variational Inequality and preserves the same safety guarantee as the original reachability formulation. Finally, inspired by the CBF-QP, we propose a QP-based online control synthesis for systems affine in control and disturbance, whose solution is always the CBVF's optimal control signal robust to bounded disturbance. We demonstrate the benefit of using the CBVFs for double-integrator and Dubins car systems by comparing it to previous methods.
\end{abstract}

\section{Introduction}
\label{sec:intro}

\subsection{Motivation \& Related Work}
\label{subsec:motivation}
Value function-based approaches are common techniques for solving safe control problems. Two such methods are Hamilton-Jacobi (HJ) reachability analysis and Control Barrier Functions (CBFs).
HJ reachability analysis formulates the reachability of a target set as an optimal control problem, and has long been used as a formal theoretical tool for safety analysis and synthesis of safe controllers \cite{lygeros2004reachability, mitchell2005}. 
HJ reachability-based value functions can be solved numerically by using the dynamic programming principle \cite{ian2005levelset}. The zero-superlevel set of the value function describes the safe set, and the optimal safety controller can be synthesized based on the gradient of the function. 
Moreover, the safe control can be robust to disturbances \cite{mitchell2005}. 

The main drawbacks of HJ reachability analysis are twofold. First, although there have been recent advances to improve computational efficiency \cite{herbert2019reachability,bansal2021deepreach, decomposition}, most numerical methods to construct the value function suffer from the curse of dimensionality \cite{hjreachabilityoverview}. 
Secondly, the resulting safe optimal control policy is generally overly conservative when applied directly.
A popular remedy for reducing conservativeness is to use a least-restrictive hybrid controller--the optimal control is only applied when the system is very close to the safe boundary. 
However, this switching-based control law often results in undesirable jerky behaviors.

\begin{figure}
\centering
\includegraphics[width=0.9\columnwidth]{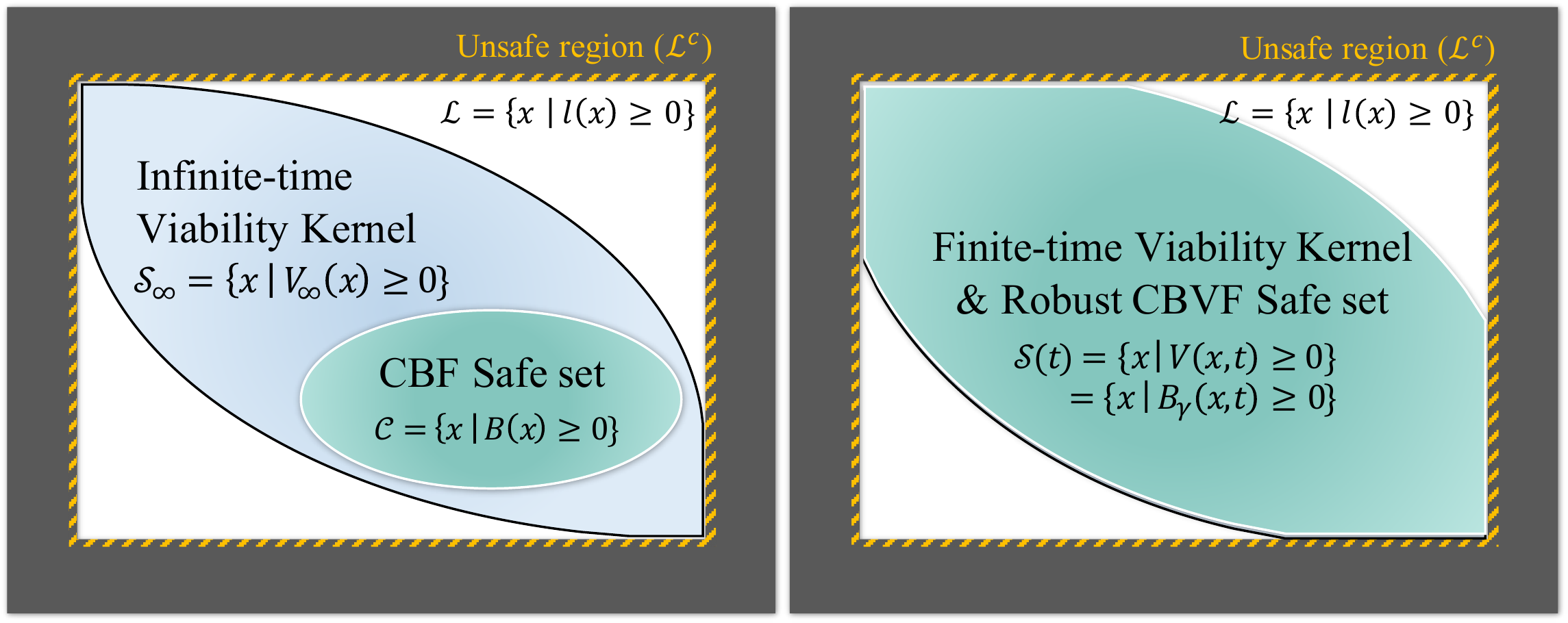}
\vspace{-.8em}
\caption{Illustrative diagram of the viability kernel and the zero-superlevel sets of the functions presented in the paper. The gray region represents the failure set that should not be entered. The constraint set $\LL$ is seen in white. On the left, the infinite-time viability kernel generated by HJ reachability is shown in blue, with the CBF safe set being a guaranteed under-approximation. On the right, the finite-time viability kernel is generally larger than the infinite-time version, as there exists more states that can be safe for only a finite time horizon. The zero-superlevel set of the proposed Robust \CBVF~$\cbf$ matches the viability kernel, the maximal robust safe set.}
\label{fig:main-diagram}
\vspace{-1.8em}
\end{figure}

Recently, Control Barrier Functions (CBFs) have gained popularity among the controls and robotics community as a convenient way of solving safe control problems \cite{WIELAND2007cbf, cbfqptac}. CBFs are Lyapunov-like functions that impose certain state-dependent constraints on the control input.
The constraint results in control invariance of the zero-superlevel set of CBFs, and this property can be used to ensure that the system stays within a desired safe region.
The main benefit of using a CBF for safety control is that for control-affine systems, the CBF constraint can be incorporated in an online min-norm optimization based controller, namely the CBF-based Quadratic Program (CBF-QP). The fact that this controller can be applied in real-time for high-dimensional systems makes it attractive for many applications \cite{nguyen2016optimal, wang17multirobot, squries2018constructive}.
Also, it can be used as an automatic safety filter (as opposed to using least-restrictive control) \cite{active_set_invariance, cheng2019rlcbf, taylor20learningcbf}.

The main drawback of CBF-based approaches is that they lack general methods of constructing a valid CBF, which results in using hand-designed or application-specific CBFs \cite{nguyen2016optimal, squries2018constructive}. 
This may restrict the system to stay only in a conservative safe region defined by a CBF's zero-superlevel set. 
Another problem arises when the system has control input bounds: the CBF may be invalid under these bounds, causing the CBF-QP to become infeasible anytime. A new QP formulation proposed recently provides pointwise feasibility but not persistent feasibility \cite{zeng2021feasibility}.

In summary, HJ reachability analysis and CBF-based safety control are complementary in many ways. HJ reachability provides constructive methods for the value functions, whereas the CBF community usually has to deal with handcrafting a valid CBF. Also, HJ-based value functions result in the maximal safe region for a desired safety constraint, whereas CBFs often can only provide a conservative estimate of safe region. On the other hand, online CBF-based safety controllers like the CBF-QP are a powerful tool to apply CBFs to high-dimensions systems in real-time applications, whereas HJ reachability suffers from the curse of dimensionality and its value function's online deployment is not straightforward due to the optimal policy's restrictive behavior. A recent paper uses HJ reachability functions as CBFs \cite{active_set_invariance}, but a theoretical understanding of the relationship between the two methods is still lacking.\vspace{-.5em}

\subsection{Paper Organization and Contributions}
\label{subsec:contributions}
In light of the fact that reachability-based value functions and CBFs are tackling a similar problem in complementary ways, we unify the two functions theoretically. First, in Sec.~\ref{sec:background} we briefly summarize and compare the concept of a value function from HJ reachability and CBFs.

In Sec.~\ref{sec:construction} we introduce the notion of a \textit{Robust Control Barrier-Value function} (\CBVF) that merges reachability-based value functions and CBFs into one function. This function (a) can be used for finite-time safety guarantees, (b) is robust to bounded disturbances, (c) recovers the maximal safe set for a desired safety constraint, and (d) leads to a safety control that satisfies the control bound everywhere inside the safe set. The main theoretical contribution is a verification that the CBVF is a viscosity solution of a particular Hamilton-Jacobi-Isaacs variational inequality (HJI-VI), and this can be used to numerically construct a valid CBVF. This constructive method does not naturally scale well, but can benefit from methods from the reachability community that enhance scalability \cite{herbert2019reachability, bansal2021deepreach, herbert2021safelearning}.

In Sec.~\ref{sec:onlinecontrol}, we introduce the optimal control policy corresponding to the CBVF. This controller is less conservative than that from the original HJ reachability, and is less jerky than using the least-restrictive controller that is commonly applied in HJ reachability. For systems affine in control and disturbance, we show that such an optimal controller can be obtained by solving a QP, namely the \textit{Robust \CBVF-QP}. In Sec.~\ref{sec:simulations}, we demonstrate this findings on numerical examples by comparing the CBVF-based safety control with the original HJ reachability and CBF-based methods.

\section{Background}

\label{sec:background}
\subsection{Problem Formulation}
\label{subsec:problem}

Consider a state trajectory of the continuous-time time-invariant controlled system with disturbance, solving
\begin{equation}
\label{eq:system}
    \dot{\trajsimple}(s) = f(\trajsimple(s), \ctrl(s), \dstb(s)),~ s\in[t,t'], \quad \text{and } \trajsimple(t)=x,
\end{equation}
where $t$ and $x$ are the initial time and state, respectively.
$\ctrl \in \cset \subset \R^{m}$ is the control input, $\dstb \in \dset \subset \R^{w}$ is the disturbance where $\cset$, $\dset$ are compact and convex sets, and $f:\R^n \times \cset \times \dset \rightarrow \R^n$ is Lipschitz continuous in the state and bounded. Let $\cfset_{[t,t']}$, $\dfset_{[t,t']}$ be a set of Lebesgue measurable functions from the time interval $[t, t']$ to $\cset$ and $\dset$, respectively. For simplicity, we set the final time as 0. For every initial time $t \le 0$, initial state $x\in\R^n$, $\ctrl(\cdot)\in\cfset_{[t,0]}$, and $\dstb(\cdot)\in\dfset_{[t,0]}$, system \eqref{eq:system} admits a unique solution trajectory. We denote this \textit{trajectory} as $\state(s)$, and will say that ``$\state(\cdot)$ solves \eqref{eq:system} for $(x, t, \ctrl, \dstb)$'' with a slight abuse of notation.

Throughout the paper, we assume that the disturbance signal $\dstb(\cdot)$ can be determined in reaction to the control signal in a form of a strategy $\strategy:\cfset_{[t,0]}\rightarrow\dfset_{[t,0]}$. However, we restrict it to draw only from \textit{nonanticipative strategies} with respect to $\ctrl(\cdot)$, denoted as $\strategy \in \Strategy_{[t,0]}$. 
The nonanticipative strategy prohibits the use of future information of the control signal to make a decision of the disturbance at each time \cite{evans_hj}.

Now consider a set $\LL$ defined as a zero-superlevel set of a bounded Lipschitz continuous function $\lf: \R^{n} \to \R$:\vspace{-.5em}
\begin{equation}
\label{eq:safeset}
\LL = \left\{x : \lf(x) \geq 0 \right\}.
\vspace{-.5em}
\end{equation}
\noindent The objective of the safety control is to guarantee the trajectory to stay in $\LL$ for $s \in [\tT, 0]$ under the worst case disturbance.
We refer to $\lf(x)$ as the \textit{safety target function}. More formally, we are interested in the following problems:

\noindent \textbullet \; \textbf{Computing the viability kernel $\Safe(t)$ \cite{lygeros2004reachability} for $\LL$:} Verify $\Safe(\tT):=\{ x \in \LL:  \forall \strategy \in \Strategy_{[\tT,0]}, \exists \ctrl(\cdot)\in\cfset_{[\tT, 0]}\; \text{s.t.}\;\forall s \in [\tT, 0], \state(s) \in \LL$
where $\state(s)$ solves \eqref{eq:system} for $(x, \tT, \ctrl, \strategy) \}$ for $\tT<0$. $\Safe(\tT)$ is the set of all the initial states at time $t$ in $\LL$ from which there exists an admissible control signal that keeps the system safe under the worst-case disturbance.

\noindent \textbullet \; \textbf{Computing a robust safe control $u(\cdot)$ for $\LL$:} For each $x \in \Safe(\tT)$, verify a control signal $\ctrl(\cdot) \in \cfset_{[\tT, 0]}$ that renders the trajectory safe for $s \in [\tT, 0]$, under the worst-case disturbance. \vspace{-.7em}

\subsection{Hamilton-Jacobi Reachability Analysis}
\label{subsec:HJReachability}
It has been verified that solving for the viability kernel and the robust safe control signal can be posed as an optimal control problem, which can be solved using HJ reachability analysis \cite{lygeros2004reachability, mitchell2005, fisac2015reach}. First, we define a cost function as\vspace{-.5em}
\begin{equation}
\label{eq:cost_brt}
    \JJ(x, t, \ctrl(\cdot), \dstb(\cdot)) := \min_{s \in [t, 0]} \lf(\state(s)),\vspace{-.5em}
\end{equation}
which captures the minimal value of $\lf(\cdot)$ along the trajectory $\state(\cdot)$ that solves \eqref{eq:system} for $(x, t, \ctrl, \dstb)$. If $\exists s \in [t, 0]$ such that $J(x, t, \ctrl(s), \dstb(s)) <0$, it means that the trajectory was violating the safety constraint at some point in the time horizon (obtaining a negative value of $\lf$), and is therefore unsafe. The objective of the safety control is to make $\JJ$ as big as possible, whereas under the worst case, the disturbance would act in a direction of decreasing $\JJ$ as much as it can. Based on this, we can define the value function $\VV:\R^n \times (-\infty, 0] \rightarrow \R$ as\vspace{-.5em}
\begin{equation}
\label{eq:value_brt}
\VV(x, t) := \min_{\strategy\in\Strategy_{[t,0]}} \max_{\ctrl \in \cfset_{[t,0]}} \JJ(x, t, \ctrl(\cdot), \strategy[\ctrl](\cdot)), \vspace{-.5em}
\end{equation}
Then, by the following proposition, the viability kernel for $\LL$ is $\Safe(\tT) = \{x \in \R^n : \VV(x, \tT) \geq 0\}$. Note that the minimum and maximum in $\Strategy_{[t,0]}, \cfset_{[t,0]}$ always exists because $U$ and $D$ are compact and convex \cite{altarovici2013general}.
 
\begin{proposition}
\label{prop:viability_kernel} For all $t\le 0$, the viability kernel for $\LL$, $\Safe(t)$, always is $\{x \in \R^n : \VV(x, t) \geq 0\}$.
\end{proposition}

\begin{proof} 
This is directly from the definition of $\VV$ and $\Safe(t)$.
\end{proof} 
 
Note that if $\Safe(\tT)$ is empty, safety can never be guaranteed under the worst-case disturbance. In the complement of $\Safe(\tT)$, the value function $\VV(x, \tT)$ is negative, therefore, for any admissible control, the trajectory is unsafe under the worst-case disturbance. This set $\Safe(\tT)^{c}$ describes what is known in the HJ Reachability community as a \textit{Backward Reachable Tube} of the unsafe set.

The value function $\VV(x, t)$ is the viscosity solution to the following Hamilton-Jacobi-Isaacs Variational Inequality (HJI-VI) \cite{fisac2015reach}:\vspace{-1.5em}

\small
\begin{align}
    0 = & \min \biggl\{\lf(x) - \VV(x, t),\biggr. \label{eq:HJI_VI_brt} \\ 
    & \left. D_{t}\VV(x, t) + \max_{\ctrl\in \cset} \min_{\dstb\in \dset} D_{x}\VV(x, t) \cdot f(x, \ctrl, \dstb) \right\} \nonumber 
\end{align}\vspace{-1em}
\normalsize

\noindent with the terminal condition $\VV(x, 0) = \lf(x)$. This means that $\VV(x,\tT)$ can be computed directly using dynamic programming backwards in time by applying the HJI-VI at each point in the state space. \vspace{-.5em}

\begin{remark}
\label{remark:viscosity}
The viscosity solution $\VV(x,t)$ is a weak solution to  \eqref{eq:HJI_VI_brt}: $\VV(x,t)$ is not differentiable for some $(x,t)$.
Under the Lipschitz assumptions for the dynamics ($f$) and the cost ($l$) in the state, $\VV(x,t)$ is Lipschitz continuous, which is differentiable almost everywhere (a.e.) in $(x,t)$-space \cite[Th.3.2.]{evans_hj}\cite{evans2015measure}.
\end{remark}

When the viability kernel $\Safe(\tT)$ is non-empty, from any element in $\Safe(\tT)$, we can synthesize a robust safe control signal from the optimal control policy. Based on whether the left or the right term in the minimum of \eqref{eq:HJI_VI_brt} is \textit{active}, the optimal policy $\pi^{*}_{V}(x, t):\R^n \times (-\infty, 0]\rightarrow U$ is determined in a different way. That is, when $\VV(x,t) < \lf(x)$,\vspace{-.5em}

\small
\begin{align}
    \label{eq:safe_policy_case_1}
    \pi^{*}_{V}(x, t) =  \argmax_{u\in U}\min_{d\in D} D_{x}V(x, t) \cdot f(x, u, d), \vspace{-.5em}
\end{align}
\normalsize
and the right term of \eqref{eq:HJI_VI_brt} is 0. Second, when $V(x, t) = l(x)$, any element of \vspace{-1em}

\small
\begin{align}
    \label{eq:safe_policy_case_2}
    K_{V}(x, t)\!:=\!\{u\!\in\!U : D_{t}V(x, t)\!+\!\min_{d\in D} D_{x}V(x, t)\!\cdot\!f(x, u, d)\!\ge\!0\}
\end{align}
\normalsize

\noindent can be used as $\pi^{*}_{V}(x, t)$. Therefore, the second case may allow multiple options for the optimal control. In either case, for any $d\in D$,\vspace{-1em}

\small
\begin{align*}
    \dot{V}(\trajsimple(t), t) &= D_{t}V(\trajsimple(t), t) \\
    &+ D_{x}V(\trajsimple(t), t) \cdot f(\trajsimple(t), \pi^{*}_{V}(\trajsimple(t), t), d) \geq 0,
\end{align*}
\normalsize

\noindent where $\state(\cdot)$ is an instantaneous trajectory of \eqref{eq:system} at $t$ with $\state(t)\!=\!x$, control $\pi_{V}^{*}(x, t)$ and disturbance $\dstb$. Therefore, this implies that for any initial state $x \in \Safe(\tT)$, for any $\strategy \in \Strategy_{[\tT, 0]}$, along the optimal trajectory $\state^{*}(\cdot)$ which solves \eqref{eq:system} for $(x, \tT, \pi_{V}^{*}, \strategy[\pi_{V}^{*}])$, the value function $\VV(\state^{*}(s), s)$ will never decrease. Since $V$ is non-negative at the initial time $\tT$, it is always kept non-negative under $\pi_{V}^{*}$ for $s\in[\tT, 0]$, which means the trajectory is rendered safe.

\begin{remark}
\label{rm:limit_of_V}
Note that for the second case of $\pi^{*}_{V}$, for any optimal $\ctrl\in K_{\VV}(x, t)$ and any $\dstb\in \dset$,\vspace{-.5em}
\begin{equation}
\label{eq:l_derivative_V}
    \dot{\lf}(\state(t))=\dot{\VV}(\state(t), t)\ge0, \vspace{-.5em}
\end{equation}
where $\state(\cdot)$ is an instantaneous trajectory of \eqref{eq:system} at $t$ with $\state(t)\!=\!x$, control $\ctrl$ and disturbance $\dstb$. This means that for the second case, $\pi^{*}_{\VV}$ requires $\lf$ to increase, in other words, it never allows the trajectory to get closer to the safety boundary. Therefore, such optimal control policy is often too restrictive to be used as a safety filter for a reference control signal. In the reachability community, to remedy this, a common practice is to switch from the reference control to the safe optimal control only when $\VV(\state(s), s)$ is close to 0, so called least-restrictive control law \cite{herbert2021safelearning, fastrack, bajcsy2019efficient}. The resulting control system with such switching law may give undesirable jerky behaviors and is prone to errors in numerically computed $D_x \VV$.\vspace{-.5em}
\end{remark}

\subsection{Control Barrier Functions}
\label{subsec:CBF}
An alternative approach for achieving the safety control objective is to use Control Barrier Functions (CBFs). The theory of CBFs is developed upon viability theory and Lyapunov-based stability theory \cite{cbfqptac}.

\begin{definition}
\label{def:cbf}
Let $\mathcal{C}$ be a zero-superlevel set of a continuously differentiable function $\BB\!:\!\R^{n}\!\to\!\R$. Consider a Lipschitz continuous controlled system without disturbance, $f\!= \!f(\state(s), \ctrl(s))$. Then $\BB$ is a \textit{Control Barrier Function} for this system if there exists an extended class $\mathcal{K}_{\infty}$ function $\alpha$ such that for all $x \in \mathcal{C}$,\vspace{-.5em}
\begin{equation}
\label{eq:cbf}
    \max_{\ctrl\in \cset}D_{x}\BB(x) \cdot f(x, \ctrl) \geq -\alpha(\BB(x)). \vspace{-.5em}
\end{equation}
\end{definition}

Introducing $-\alpha(\BB(x))$ on the right hand side of \eqref{eq:cbf} is inspired by the condition that Control Lyapunov Functions (CLFs) should satisfy in order to provide exponential stabilizability \cite{cbfqptac}. In practice, a linear function $\gamma z$ ($\gamma\!>\!0$) is often used as $\alpha(z)$. In this case, $\gamma$ serves as a \textit{maximal discount rate} of $B(\trajsimple(s))$. Informally, this means that $\BB(\trajsimple(s))$ is not allowed to decay faster than the exponentially decaying curve $\dot{\BB}(\trajsimple(s))=-\gamma \BB(\trajsimple(s))$, therefore potential unsafe behaviors smooth out as it approaches the safe boundary. More formally, the following holds: 
\begin{theorem}
\label{th:cbf}
\cite[Corollary 2]{cbfqptac} For such $\BB$ and its zero-superlevel set $\mathcal{C}$, any Lipschitz continuous controller $\pi: \mathcal{C} \to U$ such that $\pi(x) \in K_\BB (x)$ where \vspace{-.5em}
\begin{equation}
\label{eq:cbf-certificate}
    K_\BB (x) := \{\ctrl \in \cset : D_{x}\BB(x) \cdot f(x, \ctrl) \geq -\alpha(\BB(x))\}, \vspace{-.5em}
\end{equation}
will render the set $\mathcal{C}$ forward invariant \cite{cbfqptac}. In other words, $\mathcal{C}$ is control invariant. 
\end{theorem}

Condition \eqref{eq:cbf} can be incorporated in an online optimization based controller that minimizes the norm of the difference between $u$ and the reference control $u_{ref}$. For control-affine systems, this can become a Quadratic Program, namely Control Barrier Function-based Quadratic Program (CBF-QP) \cite{cbfqptac}, and can be used as an online safety filter for any reference control signal $u_{ref}$.

\subsection{Comparison between HJ reachability and CBF}

In this subsection, we restrict our interest to systems without disturbance, $f\!= \!f(\state(s), \ctrl(s))$, for the comparison between value function from the reachability $V$ and CBF $B$. Note that by extending the definition of $V$ to infinite-time horizon as $\VV_{\infty}(x)\!:=\!\lim_{t\rightarrow-\infty}\VV(x, t)$, we can get a time-invariant value function \cite{fialho_worst} whose zero-superlevel set $\Safe_{\infty} := \{ x : V_{\infty}(x) \geq 0 \}$ is a maximal control invariant set contained in $\LL$. The latter results from extending Proposition \ref{prop:viability_kernel} to infinite horizon.

The geometric connection between the zero-superlevel set of the CBF $\BB$, $\mathcal{C}$, and the zero-superlevel set of $V_{\infty}$, $\Safe_\infty$, is that $\mathcal{C}$ is always a subset of $\Safe_\infty$. This is because in order to use $\BB$ for our safety objective \eqref{eq:safeset}, the control invariant set $\mathcal{C}$ should be a subset of $\LL$, as shown in Fig.~\ref{fig:main-diagram}. Since $\Safe_\infty$ is the maximal control invariant set in $\LL$, $\mathcal{C} \subseteq \Safe_\infty$.

Also, note that $\VV_{\infty}$ satisfies the CBF condition \eqref{eq:cbf} for any extended class $\mathcal{K}_{\infty}$ function $\alpha$ where the gradient $D_x \VV_{\infty}$ exists,
from the fact, $D_t \VV_\infty =0$, and the HJI-VI \eqref{eq:HJI_VI_brt}:\vspace{-.5em}
\begin{equation*}
\max_{\ctrl\in U} D_x \VV_\infty(x)\cdot f(x,\ctrl)\geq 0 \geq -\alpha(\VV_\infty(x)). \vspace{-.5em}
\end{equation*} 
This implies that if $\VV_{\infty}$ is differentiable in $\Safe_{\infty}$, then setting $B=\VV_{\infty}$ works as a valid CBF with $\mathcal{C}=\Safe_{\infty}$. However, if it is not the case, it is hard to devise a CBF such that its zero-superlevel set recovers the maximal control invariant set in $\LL$ without relaxing its differentiability condition. Note that choosing $B=l$, which makes $\mathcal{C}=\LL$, would not be a valid CBF in general. In many cases, a valid handcrafted CBF results in its zero-superlevel set $C$ strictly smaller than $\Safe_\infty$. 
\vspace{-1em}
\section{Robust Control Barrier-Value Function and Hamilton-Jacobi-based Verification \label{sec:construction}}

Note that the condition the CBF-based safe control should satisfy, $D_{x}\BB(x) \cdot f(x, \ctrl) \geq -\alpha(\BB(x))$, from Theorem \ref{th:cbf}, is less restrictive than the condition the optimal control for $V$ should satisfy, $\min_{d\in D} D_x \VV(x, s)\cdot f(x,\ctrl,d) \ge 0$. This is mainly because of the introduction of $-\alpha(\cdot)$ on the right hand side of \eqref{eq:cbf}. Inspired by this and the fact that when $\alpha(B(x)) \equiv \gamma B(x)$, $\gamma$ serves as the maximal discount rate of $B$, we define the following new value function.

\begin{definition} A Robust Control Barrier-Value Function (\CBVF) $\BVg:\R^n \times (-\infty, 0] \rightarrow \R$ is defined as
\label{def:maximal-cbf}
\begin{equation}
\label{eq:mr-cbf-value}
    \BVg(x, t):= \min_{\strategy\in\Strategy[t, 0]} \max_{\ctrl\in\UU_{[t,0]}} \min_{s\in[t,0]} e^{\gamma(s-t)} l(\state(s)),
\end{equation}
where $\state(\cdot)$ solves for $(x, t, \ctrl, \strategy[\ctrl])$, for some $\gamma\ge0$ and $\forall t\le0$. At $t=0$, we get terminal condition $\BVg(x, 0) = \lf(x)$.
\end{definition}

Note that $\cbf$ is defined for each fixed value of $\gamma\ge0$. Now, consider the case $\gamma = 0$. For this case, the definition of $\BV_{0}$ in \eqref{eq:mr-cbf-value} matches with the definition of the original reachability-based value function in \eqref{eq:value_brt}. This is not surprising because \eqref{eq:mr-cbf-value} should be regarded as a \textit{special case of the reachability problem}, whose target function is exponentially decaying backward in time.  

Since \eqref{eq:mr-cbf-value} is an optimal control problem under a differential game setting, Bellman's principle of optimality can be applied to derive the dynamic programming principle for $\cbf$.
\vspace{-1.5em}

\begin{theorem}
\label{th:cbf-dynamic-principle} \textbf{(Dynamic Programming Optimality Condition)} For the Robust CBVF $\cbf$ in Definition \ref{def:maximal-cbf}, for each $t<t+\delta\le0$, the following is satisfied.\vspace{-1em}

\small
\begin{align}
\cbf(x, t) = \min_{\strategy\in\Strategy_{[t,0]}} \max_{u\in\UU_{[t,0]}} \min \Bigl\{ \min_{s\in[t, t+\delta]} e^{\gamma(s-t)} l(\state & (s)), \Bigr. \nonumber\\ 
\left. e^{\gamma \delta} \cbf(\state(\tpdelta),\tpdelta) \right\} & \label{eq:cbf-dp}
\end{align}
\normalsize

\noindent where $\state(\cdot)$ solves \eqref{eq:system} for $(x, t, u, \strategy)$.\vspace{-.5em}
\end{theorem}

\begin{proof}
See Appendix.\vspace{-.5em}
\end{proof}

Theorem \ref{th:cbf-dynamic-principle} leads to the derivation of the following theorem, which is the main theoretical result of this paper, showing that $\BVg$ can be obtained by solving a particular variational inequality that has the form of HJI-VI. \vspace{-.5em}

\begin{theorem}
\label{th:mr-cbf-hji-vi}
The Robust \CBVF~$\BVg$ is a Lipschitz continuous unique viscosity solution of the \CBVF~variational inequality (\CBVF-VI) below with the terminal condition $\BVg\!(x,\!0)\!=\!\lf(\!x\!)$:\vspace{-1em}

\small
\begin{align}
    & 0 = \min \biggl\{\lf(x) - \BVg(x, t),\biggr. \label{eq:HJI_VI_mr_cbf} \\ 
    & \left. D_{t}\BVg(x, t)\!+\!\max_{\ctrl\in \cset} \min_{\dstb\in \dset} D_{x}\BVg(x, t)\!\cdot\!f(x, \ctrl, \dstb)\!+\!\gamma \BVg(x, t)\!\right\} \nonumber.
\end{align}
\normalsize
\end{theorem}

\begin{proof} See Appendix.
\end{proof}\vspace{-.5em}

The following proposition shows that like the original reachability-based value function $\VV$ from \eqref{eq:value_brt}, $\BVg$ can also be used to verify the viability kernel $\Safe(\tT)$. In other words, the zero-superlevel set of the Robust CBVF contains every initial state from which robust safety guarantee is possible for a chosen time span. This is in sharp contrast to the CBFs, since the safe invariant set from a given CBF is only guaranteed to be a subset of the maximal control invariant set. Moreover, since CBVF is concerned with safety for finite-time horizon, the obtained safe set can be much bigger than the control invariant set from CBFs. Therefore, in addition to the fact that the CBVF is constructive, the main benefit of using the CBVF is that it recovers the biggest permissible region for the system for maintaining safety (Fig. \ref{fig:main-diagram}).\vspace{-.5em}

\begin{proposition}
For each $t\leq 0$, define $\mathcal{C}_{\gamma}(t):=\{x\in\R^n : \BVg(x, t) \ge 0\}$. Then, $\forall t \leq 0$, $\mathcal{C}_{\gamma}(t)=\mathcal{S}(t)$. \vspace{-.5em}

\end{proposition}
\begin{proof} For each $t\!\in\!(-\infty,0]$, consider $x$ such that $\BVg(x,t)\geq0$. For $\forall\strategy\in\Strategy_{[t,0]}$, there exists $\ctrl\in\UU_{[t,0]}$ such that $\min_{s\in[t,0]}e^{\gamma(s-t)}l(\state(s))\geq 0$.
Therefore, $x$ belongs to $\mathcal{S}(t)$.

Consider $x\in \mathcal{S}(t)$. For all $\strategy\in\Strategy_{[t,0]}$, there exists $\ctrl\in\UU_{[t,0]}$ such that $l(\state(s))$ is non-negative for all $s\in[t,0]$. 
Thus, $\max_{\ctrl\in\UU_{[t,0]}}\min_{s\in[t,0]} e^{\gamma(s-t)}l(\state(s))$ is non-negative for all $\strategy$, and $\BVg(x,t)\geq0$.
\end{proof}\vspace{-.5em}

Finally, since $\VV$ can be used to verify the viability kernel $\Safe(t)$, readers might wonder the additional benefit of introducing $\BVg$. In the next section, we explain why using $\BVg$ would be preferable to using the original value function $\VV$.

\section{Optimal Control Policy of the \CBVF}  
\label{sec:onlinecontrol}

\subsection{Evaluation of the optimal control policy of the \CBVF}
\label{subsec:onlinecontrol}\vspace{-.5em}

The main benefit of using the optimal controller from the new formulation of CBVF $\cbf$ instead of the original reachability-based optimal controller $\pi_{V}^{*}$ is that it can significantly reduce the conservativeness of $\pi_{V}^{*}$ (Remark \ref{rm:limit_of_V}).

First recall how the optimal policy $\pi_{V}^{*}$ of $\VV$ is verified: 1) when $\VV(x,t) < \lf(x)$, it is determined by \eqref{eq:safe_policy_case_1}, and 2) when $\VV(x, t) = \lf(x)$, any element of \eqref{eq:safe_policy_case_2} is optimal. 

From the \CBVF-VI \eqref{eq:HJI_VI_mr_cbf}, we can verify the optimal control policy with respect to $\BVg$ similarly. For the first case, when $\BVg(x, t) < \lf(x)$, the second term of \eqref{eq:HJI_VI_mr_cbf} must be zero; therefore the optimal control must be given by \vspace{-.5em}

\small
\begin{align}
    \label{eq:safe_policy_cbf_case_1}
    \pi^{*}_{\cbf}(x, t) =  \argmax_{\ctrl\in \cset}\min_{\dstb\in \dset} D_{x}\cbf(x, t) \cdot f(x, \ctrl, \dstb), \vspace{-.5em}
\end{align}
\normalsize
which is similar to the first case of $\pi_{V}^{*}$. Also, the CBVF-VI \eqref{eq:HJI_VI_mr_cbf} implies that for this case,\vspace{-1.5em}

\small
\begin{align}
    D_{t}\cbf(x, t) + \min_{\dstb\in \dset} D_{x}\cbf(x, t) \cdot f(x, \pi^{*}_{\cbf}(x, t), \dstb) + \gamma \cbf & (x, t)  \nonumber \\
    = \dot{\BV}_{\gamma}(\state(t), t) + \gamma \cbf(\state(t), t) = 0. \label{eq:B_derivative_case_1}&
\end{align}
\normalsize

\noindent For the second case, when $\BVg(x, t) = \lf(x)$, any element of\vspace{-1em}

\small
\begin{align}
\label{eq:safe_policy_cbf_case_2}
    K_{\cbf}(x,t)\!:=\!\{\ctrl\in \cset\!: D_{t}\BVg(x, t) + & \min_{\dstb\in \dset} D_{x}\BVg(x, t)\!\cdot\!f(x, \ctrl, \dstb) \nonumber \\ 
    & + \gamma \BVg(x, t) \ge 0\}
\end{align}
\vspace{-1em}
\normalsize

\noindent is optimal with respect to $\BVg$ and can be used as $\pi^{*}_{\BVg}$. For this case, $K_{\cbf}(x,t)$ is always non-empty because the second term of \eqref{eq:HJI_VI_mr_cbf} is greater or equal to 0, and for any $\ctrl\in K_{\cbf}(x, t)$ and any $\dstb\in \dset$,\vspace{-.5em}
\begin{equation}
\label{eq:l_derivative_B}
    \dot{\lf}(\state(t))=\dot{\BV}_{\gamma}(\state(t), t)\ge - \gamma \BVg(\state(t), t) = -\gamma \lf(\state(t)), \vspace{-.5em}
\end{equation}
where $\state(\cdot)$ solves \eqref{eq:system} for $(x,t,\ctrl,\dstb)$.

It is crucial to note the difference between \eqref{eq:l_derivative_V} and \eqref{eq:l_derivative_B}. Speaking informally, both second cases of the optimal control policies with respect to $\VV$ and $\BVg$ occur when the state is not at stake of violating safety, therefore, the user is allowed to choose any $\ctrl$ from $K_\VV$ and $K_{\BVg}$ as $\pi^{*}_{\VV}$ and $\pi^{*}_{\BVg}$, respectively. However, as Remark \ref{rm:limit_of_V} explains, $\pi^{*}_{\VV}$ still never allows the state to get closer to the safety boundary. On the other hand, $\pi^{*}_{\BV_\gamma}$ allows $\lf$ to decrease as long as it satisfies \eqref{eq:l_derivative_B}, which is a very similar property that CBFs have.
Therefore, $\pi^{*}_{\BVg}$ allows for more control authority than $\pi^{*}_{\VV}$, while achieving the same safety objective. 

This benefit of $\BVg$ over $V$ can be regarded as CBF's property of becoming less conservative instilled in the HJ reachability formulation. In the next section, we will numerically demonstrate that the optimal trajectories from $\pi^{*}_{\BVg}$ actually behave less conservative than the optimal trajectories from $\pi^{*}_{V}$, especially with higher value of $\gamma$. \vspace{-.5em}

\subsection{Online optimal policy synthesis for control-affine systems}

We end this section by proposing a specific way of synthesizing $\pi^{*}_{\BVg}$ for systems affine in control and disturbance:
\vspace{-2em}

\small
\begin{equation}
\label{eq:control-affine-sys}
    \dot{\state}(s)\!=\!f(\state(s), \ctrl(s), \dstb(s)) = p(\state(s)) + q(\state(s))\ctrl(s) + r(\state(s))d(s),
\end{equation}
\normalsize

\noindent where $p:\!\R^n\!\rightarrow\!\R^n$, $q:\!\R^n\!\rightarrow\!\R^{n\times m}$, and $r:\!\R^n\!\rightarrow\!\R^{n\times w}$.

Note that $u=\pi^{*}_{\BVg}(x, t)$ should satisfy\vspace{-.3em}
\small
\begin{equation*}
    D_{t}\BVg(x, t) + \min_{\dstb\in \dset} D_{x}\BVg(x, t)\!\cdot\!f(x, \ctrl, \dstb) \nonumber + \gamma \BVg(x, t) \ge 0 \vspace{-.3em}
\end{equation*}
\normalsize
from \eqref{eq:B_derivative_case_1} and \eqref{eq:safe_policy_cbf_case_2}. Similarly to the CBF-QP, we can incorporate this as a linear inequality constraint in a min-norm optimization based controller. When the input bound $\cset$ is polytopic, the optimization becomes a QP as well:\vspace{-.5em}

\HRule
\noindent \textbf{Robust \CBVF-QP}:
\small
\begin{subequations}
\label{eq:ftr-cbf-qp}
\begin{align}
& \pi_{QP}(x, t)= \underset{\ctrl\in \cset}{\argmin} \quad (\ctrl-\ctrl_{ref})^T (\ctrl-\ctrl_{ref})\\
& \text{s.t.}\quad  a(x,t) + D_{x}\BVg(x, t) \cdot q(x)\ctrl+ \gamma \BVg(x, t) \geq 0, \label{eq:ftr-cbf-qp-constraint}\\
& \text{where} \; a(x,t) = D_{t}\BVg(x, t) + D_{x}\BVg(x, t) \cdot p(x) \nonumber\\
& \quad  \quad \quad \quad \quad \quad + \min_{\dstb\in \dset} D_{x}\BVg(x, t) \cdot r(x)\dstb.
\end{align}
\end{subequations}
\HRule
\normalsize
\vspace{-.7em}

\noindent Note that a similar formulation is proposed in a previous work that introduces a concept of Robust CBF \cite{jankovic}. \vspace{-.5em}

\begin{proposition}
\label{prop:cbvf-qp-feasible}
For the Robust \CBVF~ $\BVg$, and for the system \eqref{eq:control-affine-sys} with linear control bound $\cset$, the Robust \CBVF-QP \eqref{eq:ftr-cbf-qp} is feasible everywhere $(x, t) \in \R^n \times (-\infty, 0]$ where the gradient $D_x \BVg(x, t)$ exists, and its solution is always an optimal policy with respect to $\BVg$. \vspace{-.5em}
\end{proposition}
\begin{proof}
For the first case, when $\BVg(x, t) < \lf(x)$, the constraint of the QP \eqref{eq:ftr-cbf-qp-constraint} is satisfied but only under the equality condition since
\vspace{-1em}

\small
\begin{equation*}
    D_{t}\BVg(x, t) + \max_{\ctrl\in \cset} \min_{\dstb\in \dset} D_{x}\BVg(x, t) \cdot f(x, \ctrl, \dstb) + \gamma \BVg(x, t) = 0
\end{equation*}
\normalsize\vspace{-1em}

\noindent from \eqref{eq:B_derivative_case_1}. Any $\ctrl\in \cset$ that satisfies the equality condition is optimal. For the second case, when $\BVg(x, t) = \lf(x)$, $K_{\BVg}$ is exactly the feasible set of the Robust \CBVF-QP.
\end{proof}

\begin{remark}
Note that any reference control signal $u_{ref}$ can be used in \eqref{eq:ftr-cbf-qp}, since Proposition \ref{prop:cbvf-qp-feasible} holds for every feasible solution. Therefore, \eqref{eq:ftr-cbf-qp} is not only an optimal controller for $\BVg$, it also can be used as a safety filter for any kind of performance controller.
As we explained in Sec. \ref{subsec:onlinecontrol}, this new safety filter is much less restrictive than the original optimal control policy of $V$. Also, compared to applying a least-restrictive safety filter explained in Remark \ref{rm:limit_of_V} which utilizes value function only at the boundary, the filter \eqref{eq:ftr-cbf-qp} can be applied globally inside $\Safe(\tT)$, and the optimization automatically adjusts $u_{ref}$ to make it safe.
\end{remark}

\begin{remark}
When the differential $D_{x}\BVg$ does not exist, since $\BVg$ is Lipschitz continuous, either one of superdifferential or subdifferential always exists. The optimal control is determined by the same rule \eqref{eq:safe_policy_cbf_case_2} where the differential $D_x \BVg(x,t)$ is replaced by the superdifferential $D_x \phi (x,t) \in D_x \BVg^{+}(x, t)$ or subdifferential $D_x \phi (x,t) \in D_x \BVg^{-}(x, t)$ \cite[Ch.3.2.5]{bardi2008optimal}. 
\end{remark}\vspace{-1em}
\section{Numerical Examples \vspace{-.5em} \label{sec:simulations}}

In the following numerical examples, standard numerical methods for computing the reachability-based value functions \cite{ian2005levelset, fisac2015reach} are used to compute $B_\gamma$.

\subsection{Double Integrator Example}\vspace{-.5em}
The running example in this subsection will be a simple 2D double integrator. Its system dynamics are
$\dot{z} = v + d,$ $\dot{v} = \ctrl,$
with states position $z$ and velocity $v$,  disturbance $\dstb \in [-0.2,0.2]$ and  control $\ctrl\in[-0.5, 0.5]$.
Figure~\ref{fig:comparison_functions} shows a comparison of the functions, zero level sets, trajectories, and control signals for three different values of $\gamma$. On the top in orange is the standard HJ VI computation (i.e. $\gamma = 0$). The other rows show computations for $\gamma = 0.2$ (middle, blue), and $\gamma = 0.5$ (bottom, cyan). 
The new formulation is also robust to bounded disturbances. Figure~\ref{fig:comparison_disturbance} shows a comparison of trajectories  under different disturbance conditions. Even under worst-case disturbances (blue), the online trajectory is guaranteed to remain in the safe set. \vspace{-.8em}

\begin{figure}\centering
\includegraphics[width=\columnwidth]{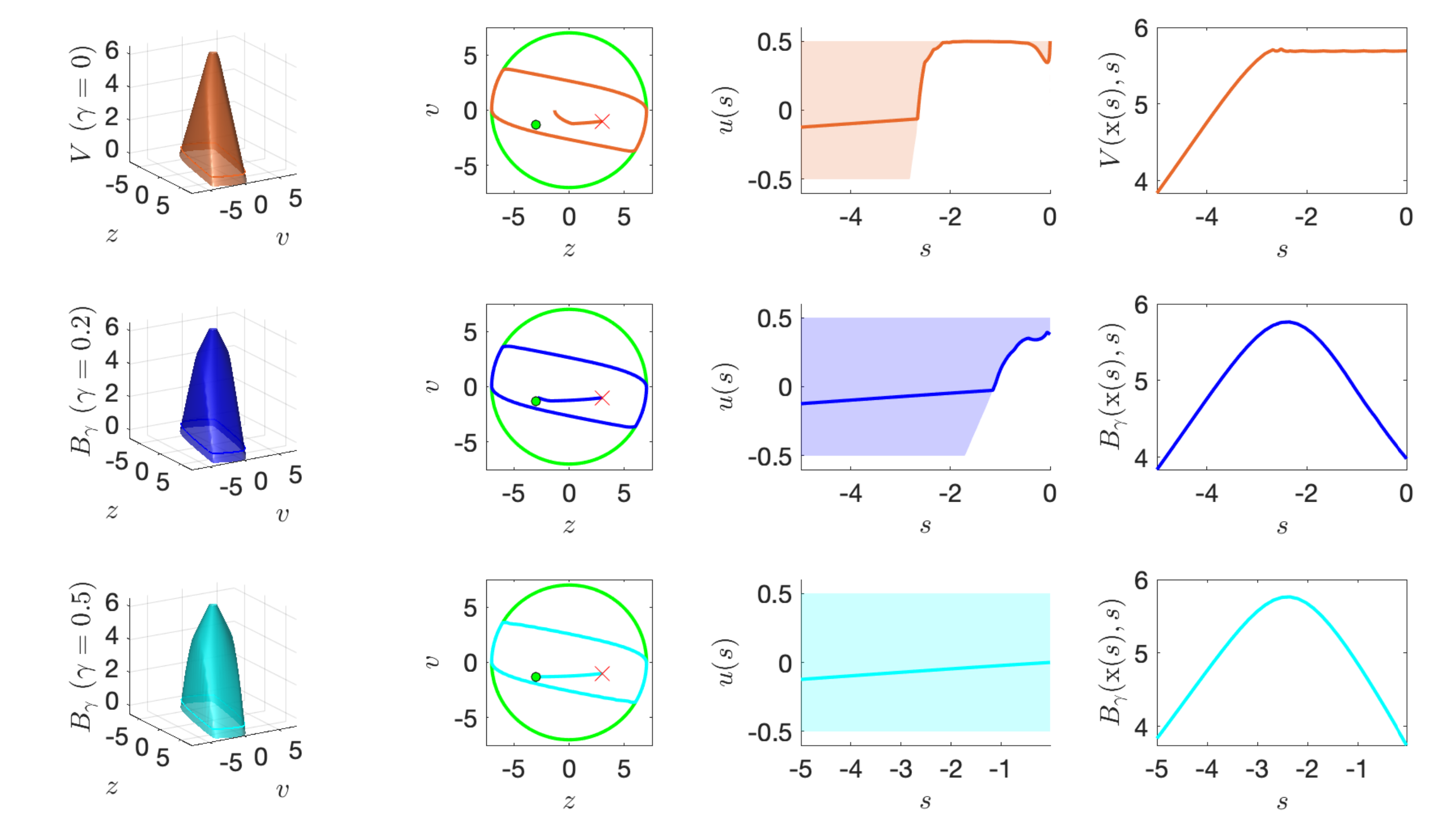}
\vspace{-2em}
\caption{
From left to right: 1. Comparison between $V(x, t)$ (top) and $\cbf(x, t)$ with $\gamma\!=\!0.2$ (middle) and $0.5$ (bottom), $t\!=\!-5$. Note that when $\gamma\!=\!0$, $\cbf(x, t)\!=\!V(x, t)$. 2. The optimal trajectories in the state space initiated at $x=[3,-1]$ (red cross) and the zero-level sets of $l(x)$ (green) and $B_\gamma(x, t)$. 3. The corresponding optimal control signals. The control is synthesized using the Robust-CBVF-QP, where $u_{ref}$ is a simple PD control for the target point (green dot), and the shaded regions indicate feasible solutions of the QP ($K_{\cbf}(\mathrm{x}(s), s)$). 4. Profiles of $B_\gamma$ along the trajectories. The optimal policy is less conservative with larger $\gamma$ (allowing $B_\gamma$ to decrease more) and is able to reach the target when $\gamma\!=\!0.5$.
}
\label{fig:comparison_functions}
\vspace{-.5em}
\end{figure}

\begin{figure}\centering
\includegraphics[width=.6\columnwidth]{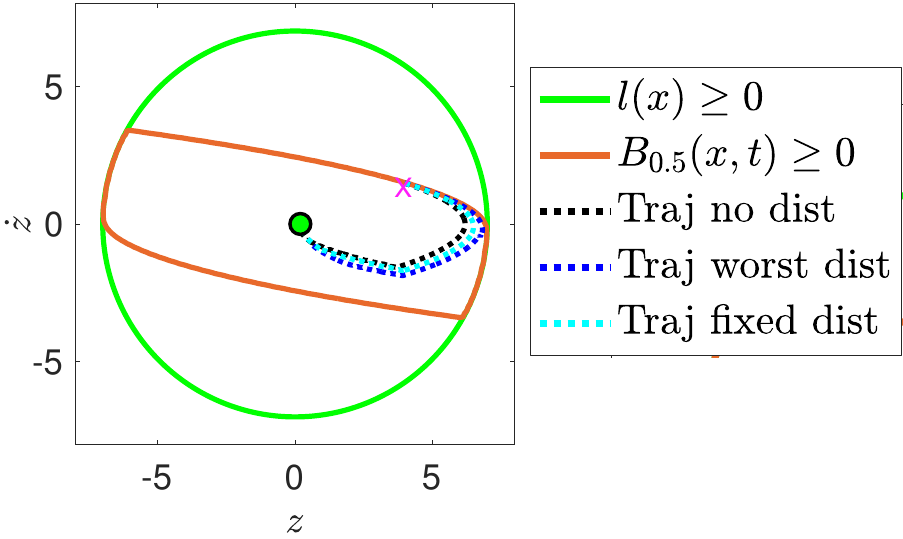} \vspace{-1em}
\caption{Trajectories under different online disturbance conditions. All trajectories start from $x=[4,1.5]$ (red cross). Conditions shown are no disturbance (black), a fixed disturbance of 0.1 m/s (cyan), and worst-case disturbance (blue). By starting in the safe set (orange boundary) the system remains within the constraint set (green boundary) even under worst-case disturbance.}
\label{fig:comparison_disturbance}
\vspace{-2em}
\end{figure}

\begin{figure}\centering
\includegraphics[width=.6\columnwidth]{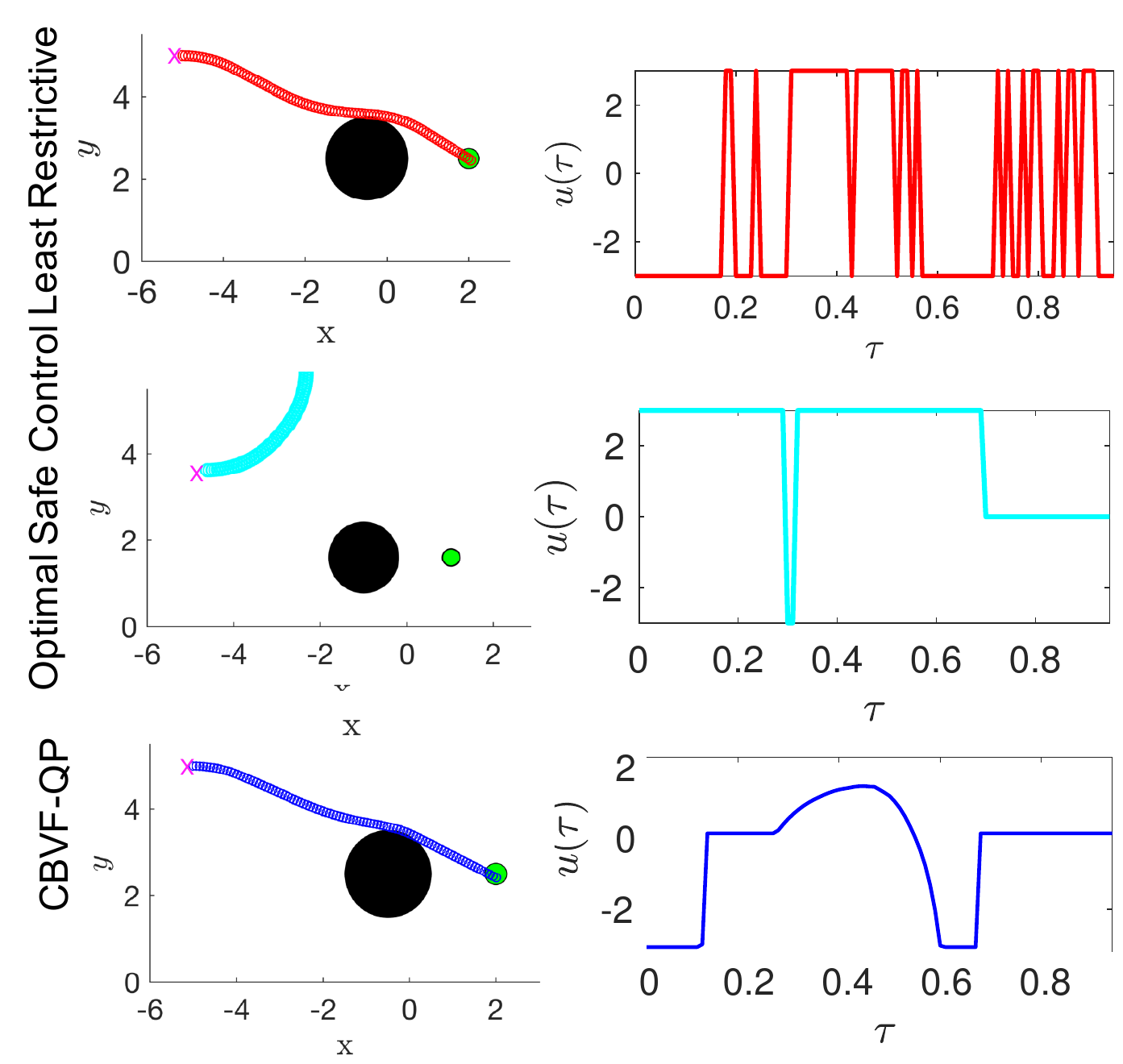}
\vspace{-1em}
\caption{Comparison between the least-restrictive controller (top), the optimal controller from the original HJ reachability (middle), and the \CBVF-QP controller (bottom).}
\label{fig:comparison_online}
\vspace{-1.2em}
\end{figure}

\begin{figure}\centering
\includegraphics[width=.8\columnwidth]{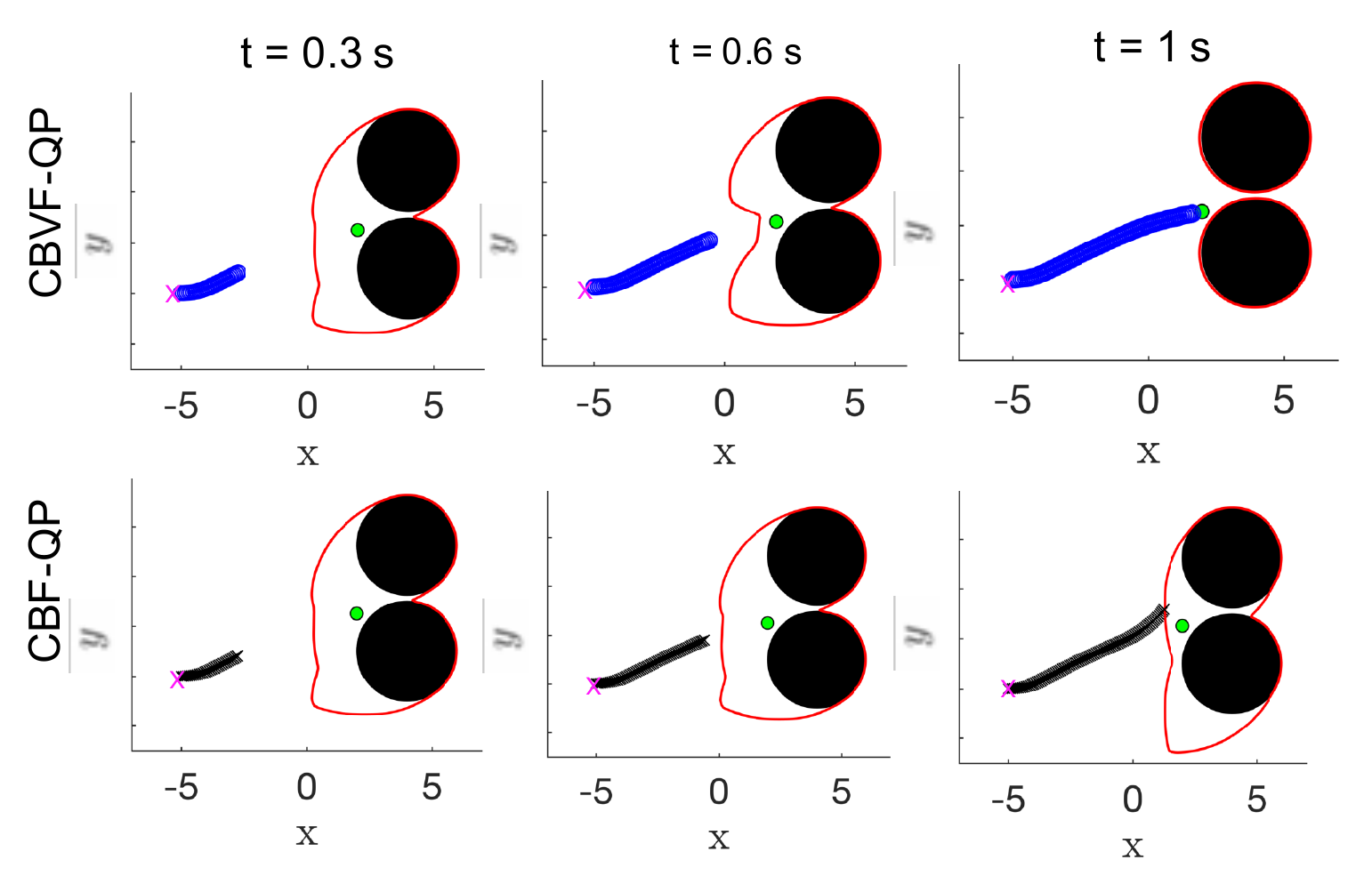}
\vspace{-1.5em}
\caption{Comparison of trajectories using a \CBVF-QP (top) vs. a CBF-QP (bottom).
The red set represents the boundary of the $x$-$y$ slice of the zero-superlevel sets of $\cbf$ and $B$ at current $\theta$. 
Note that these sets appear to rotate over time because we are visualizing the 2D slice at the current value of $\theta$. 
On the top, the system is able to reach a goal while avoiding obstacles within the prescribed time horizon.
On the bottom, the system must stay safe for an infinite time horizon, and is therefore unable to reach the goal (video: \url{https://youtu.be/wGg7rfyXCTs}).
}
\label{fig:comparison_time_varying}
\vspace{-1.8em}
\end{figure}

\subsection{Dubins Car Example}\vspace{-.3em}

In this subsection we demonstrate a comparison between using the original reachability-based controllers and the \CBVF-QP, and a comparison between the CBF-QP and the \CBVF-QP. We use a Dubins car model:
$\dot{x} = v\cos(\theta),$ $\dot{y} = v\sin(\theta),$ $\dot{\theta} = \ctrl,$
where $x,y$ are positions, $\theta$ is heading, $v$ is a fixed speed, and $\ctrl\in[-3,3]$ is rotational velocity. We use $\gamma=10$.
In Fig.~\ref{fig:comparison_online}, the system navigates around an obstacle to a goal using least-restrictive control (top) and the \CBVF-QP (bottom). The \CBVF-QP is able to use a smoother control signal and still reach the goal within the time horizon.

In Fig.~\ref{fig:comparison_time_varying}, the time stamps are shown for a system using a \CBVF-QP (which is time-varying) and a CBF-QP (which is time-invariant) controller. For the CBF, $B=\VV_{\infty}$ is used to maximize its safe set $\mathcal{C}$ as $\Safe_{\infty}$. Although $\VV_{\infty}$ has non-differentiable points for the Dubins car system in general, the trajectory resulting from the CBF-QP in Fig.~\ref{fig:comparison_time_varying} does not intersect with such points. The CBF-QP maintains safety, however, because of its safety concern for infinite-time horizon, the system is unable to reach the goal. In contrast, the time-varying \CBVF-QP allows the system to safely reach the goal within the finite-time horizon. 
This formulation can be used for scenarios that require safety only for a fixed time \cite{fixed_time_bf, ohnishi2021constraint}, for example, a hybrid system like legged robots that requires the system to stay safe only until it reaches the goal.\vspace{-.5em}

\section{Conclusion}
\label{sec:conclusion}
This paper has introduced the notion of a Control Barrier-Value Function (\CBVF) by unifying ideas from HJ reachability and Control Barrier Functions. To the best of our knowledge, this is the first constructive method for the CBF community that provides the maximal safe set for a desired safety constraint which also can handle bounded control and disturbances, however, this comes with a cost of bearing the curse of dimensionality. We also introduce the Robust \CBVF-QP for online control and demonstrate its usage as a safety filter in a double-integrator and Dubins car system. This provides a new systematic way of designing the safety filter for the reachability community. We believe the introduction of CBVFs is an important step towards bridging the gap between the CBF-based and reachability-based safety control frameworks. We plan to extend this analysis to Control Lyapunov Functions, similarly to \cite{clf_zubov}, and to reach-avoid problems \cite{fisac2015reach}.

\section*{Acknowledgements}
We would like to thank Zhichao Li, Somil Bansal, Andrea Bajcsy, and Jaime Fisac for the insightful discussions.

\appendix

The following \textbf{proofs of Theorem \ref{th:cbf-dynamic-principle} and \ref{th:mr-cbf-hji-vi}} inherit the structure from the standard proof of viscosity solution of HJI Partial Differential Equation (HJI-PDE) \cite{evans_hj}.
Note that the proofs hold for compact $U, D$ without convexity condition. Here, we use notation $\traj_{x, t}^{u,d}\equiv\state:[t, 0]\!\rightarrow\!\R^n$, where $\state(\cdot)$ solves \eqref{eq:system} for $(x, t, u, d)$, instead of $\state(\cdot)$, to specify control and disturbance signal. We use $\Strategy_{t} := \Strategy_{[t, 0]}$, $\UU_{t}:=\UU_{[t, 0]}$.

\begin{proof} \textbf{(Proof of Theorem \ref{th:cbf-dynamic-principle})}

Define $W(x, t)$ as the right hand side of \eqref{eq:cbf-dp}. For $\forall \epsilon > 0$, $\exists \eta\in\Strategy_t$ such that $\forall u\in\UU_t$ \vspace{-1em}

\small
\begin{align}
W(x, t) \ge \min \Bigl\{ \inf_{s\in[t, t+\delta]} e^{\gamma(s-t)} l(\traj_{x,t}^{u,\eta[u]}(s)), \Bigr. \nonumber & \\ 
\left. e^{\gamma \delta}  \cbf(\traj_{x,t}^{u,\eta[u]}(\tpdelta),\tpdelta)\right\}& -\epsilon \label{eq:dp-cond1}
\end{align}
\normalsize

\noindent From the definition of $\cbf$, for each $y\in\R^n$,\vspace{-1em}

\small
\begin{equation*}
    \cbf(y, \tpdelta)= \inf_{\strategy\in\Strategy_{t+\delta}} \sup_{u\in\UU_{t+\delta}} \inf_{s\in[\tpdelta,0]} e^{\gamma(s-(\tpdelta))} l(\traj_{y,\tpdelta}^{u,\strategy[u]}(s)).
\end{equation*}
\normalsize

\noindent Therefore, $\forall \epsilon_1 > 0$, $\exists \eta_y \in\Strategy_{\tpdelta}$ such that for $\forall u\in\UU_{\tpdelta}$ \vspace{-1em}

\small
\begin{equation}
\label{eq:dp-cond2}
    \cbf(y, \tpdelta) \ge \inf_{s\in[\tpdelta,0]} e^{\gamma(s-(\tpdelta))}  l(\traj_{y,\tpdelta}^{u,\eta_{y}[u]}(s)) - \epsilon_1.
\end{equation}
\normalsize

\noindent With $y := \traj_{x, t}^{u, \eta[u]}(\tpdelta)$, define \vspace{-0.5em}

\small
\begin{equation*}
    \strategy[u] := \begin{cases} \eta[u](s) &\mbox{for}\quad t\le s \le \tpdelta \\ \eta_{y}[u](s) &\mbox{for}\quad \tpdelta < s \le 0\end{cases}
\end{equation*}
\normalsize

\noindent Then, from \eqref{eq:dp-cond1} and \eqref{eq:dp-cond2}, for $\forall u \in \UU_t$, \vspace{-1em}

\small
\begin{align*}
W(x, t)\!\ge\!\inf_{s\in[t,0]} e^{\gamma(s-t)} l(\traj_{x,t}^{u,\strategy[u]}(s))\!-\!2 \epsilon \;\; \text{by taking }\epsilon_1\!=\!e^{-\gamma\delta}\epsilon.
\end{align*}
\normalsize

\noindent Therefore, \vspace{-.5em}
\small
\begin{equation}
\label{eq:dp-cond3}
W(x, t) \ge \cbf(x, t) - 2 \epsilon\quad \forall \epsilon > 0.
\end{equation}
\normalsize

On the other hand, by definitions of $\cbf$ and $W$, for $\forall \epsilon > 0$, $\exists \eta \in \Strategy_t$ such that for $\forall u \in \UU_t$ \vspace{-1em}

\small
\begin{flalign}
    & \inf_{s\in[t, 0]} e^{\gamma(s-t)} l(\traj_{x,t}^{u,\eta[u]}(s)) \le \cbf(x, t) + \epsilon \label{eq:dp-cond4}
\end{flalign}
\vspace{-1em}
\begin{alignat}{2}
    & W(x, t) \le \sup_{u\in\UU_t} \min \Bigl\{ && \inf_{s\in[t, t+\delta]} e^{\gamma(s-t)} l(\traj_{x,t}^{u,\eta[u]}(s)), \Bigr. \nonumber \\ 
& && \left. e^{\gamma \delta}  \cbf(\traj_{x,t}^{u,\eta[u]}(\tpdelta),\tpdelta)\right\}, \nonumber
\vspace{-.5em}
\end{alignat}
\normalsize
Therefore, $\exists u_0\in \UU_t$ such that
\small
\begin{align}
    W(x, t) \le \min \Bigl\{ \inf_{s\in[t, t+\delta]} e^{\gamma(s-t)} l(\traj_{x,t}^{u_0,\eta[u_0]}(s)), \Bigr. \nonumber & \\ 
\left. e^{\gamma \delta}  \cbf(\traj_{x,t}^{u_0,\eta[u_0]}(\tpdelta),\tpdelta)\right\}&\!+ \!\epsilon \label{eq:dp-cond5}
\vspace{-1em}
\end{align}
\normalsize

\noindent For $\forall u \in \UU_{\tpdelta}$, define \vspace{-.5em}

\small
\begin{equation*}
    \bar{u}(s) := \begin{cases} u_0(s) &\mbox{for}\quad t\le s \le \tpdelta \\ u(s) &\mbox{for}\quad \tpdelta < s \le 0,\end{cases}
\end{equation*}
\normalsize

\noindent and define $\bar{\eta}\in\Strategy_{\tpdelta}$ by $\bar{\eta}[u](s)=\eta[\bar{u}](s)$ for $s\in[\tpdelta, 0]$. Then, with $y:=\traj_{x, t}^{\bar{u}, \eta[\bar{u}]}(\tpdelta)$, by definition of $\cbf$, \vspace{-1em}

\small
\begin{equation*}
    \cbf(y, \tpdelta)\!\le\!\sup_{u\in\UU_{t+\delta}} \inf_{s\in[\tpdelta,0]} e^{\gamma(s-(\tpdelta))}  l(\traj_{y,\tpdelta}^{u,\bar{\eta}[u]}(s)).
\end{equation*}
\normalsize

\noindent Therefore, $\forall \epsilon_2 > 0$, $\exists u_1 \in \UU_{\tpdelta}$ such that \vspace{-1em}

\small
\begin{equation}
\label{eq:dp-cond6}
    \cbf(y, \tpdelta)\!\le\! \inf_{s\in[\tpdelta,0]} e^{\gamma(s-(\tpdelta))}  l(\traj_{y,\tpdelta}^{u_1,\bar{\eta}[u_1]}(s)) + \epsilon_2.
\end{equation}
\normalsize

\noindent Selecting $\bar{u}\in\UU_t$ with $u_1 \in \UU_{\tpdelta}$, \eqref{eq:dp-cond5} and \eqref{eq:dp-cond6} yields \vspace{-1em}

\small
\begin{align*}
W(x, t) 
\le \inf_{s\in[t,0]} e^{\gamma(s-t)} l(\traj_{x,t}^{u,\eta[u]}(s))\!+\!2 \epsilon \;\; \text{by taking }\epsilon_2\!=\!e^{-\gamma\delta}\epsilon.
\end{align*}
\normalsize

\noindent Therefore, from \eqref{eq:dp-cond4}, \vspace{-.5em}
\small
\begin{equation}
\label{eq:dp-cond7}
W(x, t) \le \cbf(x, t) + 3 \epsilon\quad \forall \epsilon > 0.
\vspace{-.5em}
\end{equation}
\normalsize

\noindent The proof is done from \eqref{eq:dp-cond3} and \eqref{eq:dp-cond7}.
\end{proof} 

\noindent \textit{Proof.} \textbf{(Proof of Theorem \ref{th:mr-cbf-hji-vi})}

According to the definition of viscosity solution \cite{bardi2008optimal}, Theorem \ref{th:mr-cbf-hji-vi} is equivalent to $\cbf$ satisfying the following statements.

\noindent 1) For $\forall \phi (x, t) \in C^{1} (\R^n \times (-\infty, 0])$ such that $\cbf-\phi$ has a \textit{local maximum} 0 at $(\xntn)\in \R^n \times (-\infty, 0]$,
\vspace{-1em}

\small
\begin{align}
    & 0 \le \min \biggl\{l(x_0) - \phi(\xntn),\biggr. \label{eq:cbf-viscosity-1} \\ 
    & \left. D_{t}\phi(\xntn) \!+\! \max_{u\in U} \min_{d\in D} D_{x}\phi(\xntn) \cdot f(x_0, u, d) \!+\! \gamma\phi(\xntn)\right\}. \nonumber
\end{align}
\normalsize

\noindent 2) For $\forall \phi (x, t) \in C^{1} (\R^n \times (-\infty, 0])$ such that $\cbf-\phi$ has a \textit{local minimum} 0 at $(\xntn)\in \R^n \times (-\infty, 0]$,
\vspace{-1em}

\small
\begin{align}
    & 0 \ge \min \biggl\{l(x_0) - \phi(\xntn),\biggr. \label{eq:cbf-viscosity-2} \\ 
    & \left. D_{t}\phi(\xntn) \!+\! \max_{u\in U} \min_{d\in D} D_{x}\phi(\xntn) \cdot f(x_0, u, d) \!+\! \gamma\phi(\xntn)\right\}. \nonumber
\end{align}
\normalsize

\vspace{-.5em}

We use the following lemma to prove 1) and 2).
\begin{lemma}
\label{lm:apdx_1}
For $\phi (x, t) \in C^{1} (\R^n \times (-\infty, 0])$, define
\vspace{-1em}

\small
\begin{equation}
    \Lambda_{\phi}(x, t, u, d)\!:=\!D_t \phi(x, t) + \!D_x \phi(x, t)\cdot\!f(x, u, d)+\!\gamma \phi(x, t).
\end{equation}\normalsize

\noindent (a) If $\exists \theta >0$, $\exists (x_0, t_0) \in \R^n \times (-\infty, 0]$ such that $\max_{u\in U} \min_{d \in D}\Lambda_{\phi}(\xntn, u, d) \le -\theta$, there exists a small enough $\delta>0$, $\exists \strategy \in \Strategy_{t_0}$ such that $\forall u\in \UU_{t_0}$,
\small
\begin{equation}
\label{eq:lm_1_1}
    e^{\gamma\delta}  \phi(\traj_{\xntn}^{u,\strategy[u]}(\topdelta), \topdelta)-\phi(\xntn)\!\le\!- \frac{\theta}{2} \delta.
\end{equation}
\normalsize

\noindent (b) If $\exists \theta >0$, $\exists (x_0, t_0) \in \R^n \times (-\infty, 0]$ such that $\max_{u\in U} \min_{d \in D}\Lambda_{\phi}(\xntn, u, d) \ge \theta$, there exists a small enough $\delta>0$, $\forall \strategy \in \Strategy_{t_0}$, $\exists u \in \UU_{t_0}$ such that
\small
\begin{equation}
\label{eq:lm_1_2}
    e^{\gamma\delta}  \phi(\traj_{\xntn}^{u,\strategy[u]}(\topdelta), \topdelta)\!-\!\phi(\xntn) \ge \frac{\theta}{2} \delta.
\end{equation}
\end{lemma}
\normalsize

\noindent Lemma \ref{lm:apdx_1} is a modification of \cite[Lemma 4.3.]{evans_hj} for general HJI-PDE to \CBVF-VI. For its proof, please refer to \cite{evans_hj}.

\vspace{0.5em}
\noindent \textit{Proof of 1).} Let \eqref{eq:cbf-viscosity-1} be false. Then one of the followings should hold. \vspace{-0.5em}
\small
\begin{subequations}
\begin{align}
    \exists \theta_1 > 0\;\text{s.t.}\;\; & l(x_0) - \phi(\xntn) \le -\theta_1 \label{eq:main-1-1}\\
    \exists \theta_2 > 0\;\text{s.t.}\;\; & D_t \phi(\xntn)\!+\!\max_{u\in U} \min_{d\in D} D_x \phi(\xntn)\!\cdot\!f(x_0, u, d)  \nonumber \\
    & + \gamma \phi(\xntn) \le -\theta_2 \label{eq:main-1-2}
\end{align}
\end{subequations}
\normalsize

If \eqref{eq:main-1-1} is true, by continuity of $l$ in the state and $\traj$ in time, $\exists \delta > 0$ such that for all $u \in \UU_{t_0}$, $\strategy \in \Strategy_{t_0}$, $s\in[t_0, t_0 + \delta]$, \vspace{-.5em}
\small
\begin{equation*}
    \left| e^{\gamma(s-t_0)} l(\traj_{\xntn}^{u, \strategy[u]}(s)) - l(x_0) \right| \le \frac{\theta_1}{2}. \vspace{-.5em}
\end{equation*}
\begin{align*}
    \Rightarrow e^{\gamma(s-t_0)} l(\traj_{\xntn}^{u, \strategy[u]}(s)) & \le l(x_0) + \frac{\theta_1}{2} \\
    & \le \phi(\xntn)\!-\!\frac{\theta_1}{2} = \cbf(\xntn)\!-\!\frac{\theta_1}{2}.
\end{align*}

\normalsize
\noindent Plugging this into the dynamic programming principle \eqref{eq:cbf-dp},
\small
\begin{align*}
    \cbf(\xntn) \le & \inf_{\strategy\in\Strategy_{t_0}} \sup_{u\in\UU_{t_0}} \inf_{s \in [t_0,t_0+\delta]} e^{\gamma(s-t_0)}  l(\traj_{\xntn}^{u, \strategy[u]}(s)) \\
    \le & \;\; \cbf(\xntn) - \frac{\theta_1}{2}.
\end{align*}
\normalsize
This is a contradiction, therefore, \eqref{eq:main-1-1} is false.

Next, if \eqref{eq:main-1-2} is true, from Lemma \ref{lm:apdx_1}.a, for small enough $\delta > 0$, $\exists \eta \in \Strategy_{t_0}$ such that for all $u\in \UU_{t_0}$,
\small
\begin{equation*}
    e^{\gamma\delta}  \phi(\traj_{\xntn}^{u,\eta[u]}(\topdelta), \topdelta)-\phi(\xntn)\!\le\!- \frac{\theta_2}{2} \delta.
\end{equation*}
\normalsize
Since $\cbf-\phi$ has local maximum 0 at $(\xntn)$,
\small
\begin{equation*}
    \cbf(\traj_{\xntn}^{u,\eta[u]}(t_{0}\!+\!\delta),\topdelta) - \phi(\traj_{\xntn}^{u,\eta[u]}(\topdelta),\topdelta) \le 0. \vspace{-1em}
\end{equation*}
\begin{align*}
    \Rightarrow \! e^{\gamma\delta} \cbf(\traj_{\xntn}^{u,\eta[u]}&(\topdelta),\topdelta) \le e^{\gamma\delta} \phi(\traj_{\xntn}^{u,\eta[u]}(\topdelta),\topdelta) \nonumber \\
    &\le \phi(\xntn) - \frac{\theta_2}{2}\delta \nonumber = \cbf(\xntn)  - \frac{\theta_2}{2}\delta.
\end{align*}
\normalsize
Finally, from \eqref{eq:cbf-dp}, we get, \vspace{-1em}

\small
\begin{alignat*}{2}
\cbf(\xntn) \le & \sup_{u\in\UU_{t_0}} \min \Bigl\{\inf_{s\in[t_0, \topdelta]} e^{\gamma(s-t_0)} l(\traj_{\xntn}^{u,\eta[u]}(s)), \Bigr.&& \\ 
& \quad \quad \quad \quad \quad \left. e^{\gamma \delta}  \cbf(\traj_{\xntn}^{u,\eta[u]}(\topdelta),\topdelta) \right\}&& \\
\le & \cbf(\xntn) - \frac{\theta_2}{2}\delta,
\end{alignat*}
\normalsize

\noindent which is a contradiction. Therefore, \eqref{eq:main-1-2} is false. \qed
\vspace{.5em}
\noindent \textit{Proof of 2).} Let \eqref{eq:cbf-viscosity-2} be false. Then both of the followings should hold. \vspace{-0.5em}
\small
\begin{subequations}
\begin{align}
    \exists \theta_1 > 0\;\text{s.t.}\;\; & l(x_0) - \phi(\xntn) \ge \theta_1 \label{eq:main-2-1}\\
    \exists \theta_2 > 0\;\text{s.t.}\;\; & D_t \phi(\xntn)\!+\!\max_{u\in U} \min_{d\in D} D_x \phi(\xntn)\!\cdot\!f(x_0, u, d)  \nonumber \\
    & + \gamma \phi(\xntn) \ge \theta_2 \label{eq:main-2-2} \vspace{-0.5em}
\end{align}
\end{subequations}
\normalsize

From \eqref{eq:main-2-1}, by continuity of $l$ and $\traj$, $\exists \delta_1 > 0$ such that for all $u \in \UU_{t_0}$, $\strategy \in \Strategy_{t_0}$, $s\in[t_0, t_0 + \delta_1]$, \vspace{-0.5em}
\small
\begin{equation*}
    \left| e^{\gamma(s-t_0)}  l(\traj_{\xntn}^{u, \strategy[u]}(s)) - l(x_0) \right| \le \frac{\theta_1}{2}.
    \vspace{-0.5em}
\end{equation*}
\begin{align}
    \Rightarrow e^{\gamma(s-t_0)} l(\traj_{\xntn}^{u, \strategy[u]}(s)) & \ge l(x_0) - \frac{\theta_1}{2} \ge \phi(\xntn)\!+\!\frac{\theta_1}{2} \nonumber \\
    & = \cbf(\xntn)\!+\!\frac{\theta_1}{2}. \label{eq:main-2-cond1}
\end{align}

\normalsize
\noindent From \eqref{eq:main-2-2}, by Lemma \ref{lm:apdx_1}.b, for small enough $\delta_2 > 0$, $\forall \strategy \in \Strategy_{t_0}$, $\exists u_2 \in \UU_{t_0}$ such that 
\small
\begin{equation}
    e^{\gamma\delta_2}  \phi(\traj_{\xntn}^{u_2,\strategy[u_2]}(t_0\!+\!\delta_2), t_0\!+\!\delta_2)\!-\!\phi(\xntn) \ge \frac{\theta_2}{2} \delta_2.
\end{equation}
\normalsize

\noindent Since $\cbf-\phi$ has a local minimum 0 at $(\xntn)$, \vspace{-1em}

\small
\begin{align}
    e^{\gamma\delta_2} \cbf(&\traj_{\xntn}^{u_2,\strategy[u]}  (t_0\!+\!\delta_2), t_0\!+\!\delta_2) \nonumber \\ 
    &\ge e^{\gamma\delta_2} \phi(\traj_{\xntn}^{u_2,\strategy[u]}(t_0\!+\!\delta_2),t_0\!+\!\delta_2) \nonumber \\
    &\ge \phi(\xntn) + \frac{\theta_2}{2}\delta_2 = B(\xntn) + \frac{\theta_2}{2}\delta_2. \label{eq:main-2-cond2}
\end{align}
\normalsize

\noindent Take $\delta = \min(\delta_1, \delta_2)$ and plugging \eqref{eq:main-2-cond1} and \eqref{eq:main-2-cond2} into \eqref{eq:cbf-dp}, \vspace{-1em}

\small
\begin{alignat*}{2}
\cbf(\xntn) \ge & \inf_{\strategy\in\Strategy_{t_0}} \min \Bigl\{\inf_{s\in[t_0, \topdelta]} e^{\gamma(s-t_0)} l(\traj_{\xntn}^{u_2,\strategy[u_2]}(s)), \Bigr.&& \\ 
& \quad \quad \quad \quad \quad \left. e^{\gamma \delta}  \cbf(\traj_{\xntn}^{u_2,\strategy[u_2]}(\topdelta),\topdelta) \right\}&& \\
\ge & \;\; \cbf(\xntn) +\min \left\{ \frac{\theta_2}{2}\delta, \frac{\theta_1}{2} \right\}.
\end{alignat*}
\normalsize

\noindent This is a contradiction. \qed

Note that since $\cbf$ satisfies both 1) and 2), uniqueness and Lipschitz continuity of $\cbf$ can be derived similarly to \cite[Th.4.2]{barron1989bellman} and \cite[Th.3.2.]{evans_hj}, respectively.\vspace{-.5em}





\bibliographystyle{./IEEEtran} 
\bibliography{./IEEEabrv,./IEEEexample}

\end{document}